\documentclass[letterpaper, 10 pt, conference]{ieeeconf}  

\IEEEoverridecommandlockouts                              
\usepackage{graphics} 
\usepackage{epsfig} 
\usepackage{mathptmx} 
\usepackage{times} 
\usepackage{amsmath} 
\usepackage{amssymb}  
\usepackage{amsthm}
\usepackage{algorithm}
\usepackage{algpseudocode}
\usepackage{xcolor}

\theoremstyle{plain}
\newtheorem{theorem}{Theorem}
\newtheorem*{theorem*}{Theorem}

\newcommand{\PP}{\mathbf{P}}
\newcommand{\A}{\mathbf{A}}
\newcommand{\G}{\mathbf{G}}
\newcommand{\X}{\mathbf{X}}
\newcommand{\Z}{\mathbf{Z}}
\newcommand{\E}{\mathbf{E}}
\newcommand{\T}{\mathbf{T}}

\DeclareMathOperator*{\argmin}{arg\,min}

\title{\LARGE \bf
Trial-and-Error Learning in Decentralized Matching Markets
}

\author{Vade Shah, Bryce L. Ferguson, and Jason R. Marden 
\thanks{This work is supported by ONR grant \#N00014-20-1-2359, AFOSR grants \#FA9550-20-1-0054 and \#FA9550-21-1-0203, and NSF GRFP grant \#2139319.}
\thanks{V. Shah ({\tt\small vade@ucsb.edu}) and J. R. Marden are with the Department of Electrical and Computer Engineering at the University of California, Santa Barbara, CA. B. L. Ferguson is with the Department of Electrical Engineering and Computer Sciences at the University of California, Berkeley, CA.}%
}

\begin{document}

\maketitle
\thispagestyle{empty}
\pagestyle{empty}

\begin{abstract} Two-sided matching markets, environments in which two disjoint groups of agents seek to partner with one another, arise in several contexts. In static, centralized markets where agents know their preferences, standard algorithms can yield a stable matching. However, in dynamic, decentralized markets where agents must learn their preferences through interaction, such algorithms cannot be used. Our goal in this paper is to identify achievable stability guarantees in decentralized matching markets where (i) agents have limited information about their preferences and (ii) no central entity determines the match. Surprisingly, our first result demonstrates that these constraints do not preclude stability—simple ``trial and error" learning policies guarantee convergence to a stable matching without requiring coordination between agents. Our second result shows that more sophisticated policies can direct the system toward a particular group's optimal stable matching. This finding highlights an important dimension of strategic learning: when agents can accurately model others' policies, they can adapt their own behavior to systematically influence outcomes in their favor—a phenomenon with broad implications for learning in multi-agent systems.
\end{abstract}


\section{Introduction}

Two-sided matching markets are a fundamental feature of various socioeconomic and engineered systems in which agents from two distinct groups interact to form mutually beneficial partnerships. These markets are ubiquitous, but the ways in which matchings---sets of partnerships between agents---emerge can vary significantly across contexts. In the American residency admissions process, for instance, residency applicants and hospitals submit their rankings of one another to the National Resident Matching Program (NRMP) (Figure \ref{fig:examples}, left), which uses these rankings to identify a stable matching where no applicant and hospital prefer one another to their assigned partner. Stability is paramount, as it ensures notions of fairness and efficiency in assignments. At least one stable matching is guaranteed to exist in every two-sided matching market \cite{gale1962college}, and several celebrated algorithms can identify them \cite{gale1962college, vate1989linear, roth1993stable, roth1997effects}. Because this process is \emph{static} (hospitals and residents form one, permanent matching) and because the market is \emph{centralized} (the matching is assigned), stability is always ensured.

On the other hand, in \emph{dynamic, decentralized} matching markets, agents repeatedly form impermanent partnerships without a centralized matchmaker. Oftentimes, matchings in these markets arise as the result of an active group of agents proposing to a passive group that accepts these proposals. Depending on the agents' policies—how they choose to make and accept proposals—different matchings may arise over time. When agents know their preferences, natural policies ensure eventual convergence to a stable matching \cite{gale1962college, roth1990random, blum1997vacancy, ackermann2008uncoordinated}, indicating that the stability engineered in centralized markets often arises organically in decentralized markets.

However, in many modern decentralized markets, agents do not know their own preferences \textit{a priori}. For example, consider Amazon Mechanical Turk (MTurk), an online crowdsourcing platform where employers seek out workers to complete short-term contracts. MTurk hosts thousands of regular users on both sides of the market, which can make it difficult for workers to accurately assess and compare all possible employers and vice versa. However, once a worker accepts an offer from an employer, they both learn key attributes about one another (e.g., pay, quality of work), that informs their decisions about future employment opportunities. Importantly, in markets like MTurk, this kind of \emph{information}—or the lack thereof—dictates not just how agents follow policies, but also what a policy is, requiring innovations for ensuring convergence to stability.

Several recent works have designed policies whereby agents \emph{learn} to form stable matchings through repeated interactions and observations in dynamic, decentralized markets. In the case of \emph{one-sided} learning, where the accepting group of agents knows their own preferences \emph{a priori}, but the proposing group must learn them over time, a large body of work has extended multi-armed bandits results \cite{lai1985asymptotically, bubeck2012regret} to develop policies that guarantee the proposers' regret grows at most logarithmically \cite{liu2020competing, liu2021bandit, basu2021beyond, maheshwari2022decentralized, kong2023player, hosseini2024putting}. Recently, these results have also been extended to the case of \emph{two-sided} learning (Figure \ref{fig:examples}, right) where both groups of agents must learn their own preferences \cite{das2005two, pagare2023two, pokharel2023converging}.


\begin{figure}
    \centering
    \includegraphics[width=\linewidth]{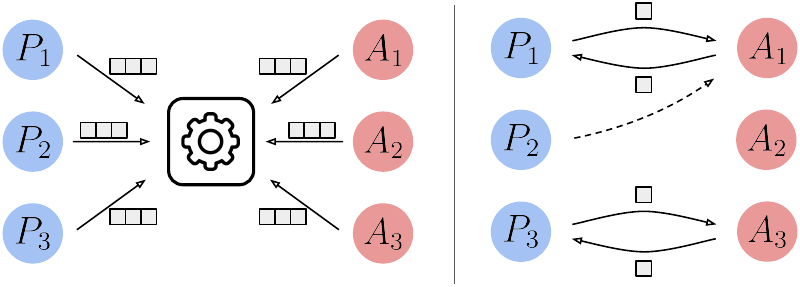}
    \caption{A two-sided matching market with 3 proposers and 3 acceptors. In a static, centralized market with complete information (left), an algorithm gathers agents' known preferences and assigns a matching. In a dynamic, 
    decentralized environment with two-sided uncertainty (right), agents learn their own preferences over one another as they interact and form matchings.}
    \label{fig:examples}
\end{figure}

Although these bandit-based techniques provide valuable transient guarantees, their long-run convergence guarantees are limited. In contrast, a parallel line of work on the one-sided problem has utilized techniques from learning in games \cite{young1993evolution, marden2009payoff, pradelski2012learning, marden2014achieving}, a subfield of game theory concerned with the design of simple policies that ensure convergence to desirable configurations (Nash equilibria, correlated equilibria, etc.), to guarantee convergence to stable matchings in decentralized markets \cite{etesami2024decentralized, shah2024learning}. While these simple policies do not provide strong transient guarantees, they ensure convergence, enhancing our understanding of what is fundamentally achievable with learning approaches.

Motivated by the latter direction, this paper asks whether \textbf{there exist policies that guarantee convergence to stable matchings in two-sided decentralized matching markets where no agent knows their own preferences}. Our first result (Theorem \ref{thm:sm}) demonstrates that stability is still achievable even in the absence of information and a central matchmaker: simple ``trial-and-error" learning policies that only use an agent's own prior observations guarantee convergence to a stable matching. Our second result (Theorem \ref{thm:aosm}) delves further into what can be achieved with learning, showing that when the proposers follow simple trial-and-error learning while the acceptors follow a more sophisticated policy, they converge instead to the acceptors' optimal stable matching. These results indicate that, fundamentally, little is needed to achieve stable matchings; moreover, if one group of agents can model others' policies, then they can further ensure that the equilibrium outcome lies in their favor.


\section{Model}

In this section, we first introduce the basic elements of the two-sided matching market and discuss its equilibrium concepts. Then, we present a game-theoretic model of how agents repeatedly form matchings and learn from their observations in a dynamic, decentralized matching market.

\subsection{Markets and Matchings}

We describe a two-sided matching market by the tuple $(\PP, \A, \mathbf{U})$, where
\begin{itemize}
    \item $\PP \triangleq \{ P_1, \dots, P_n \}$ is the set of \emph{proposers},
    \item $\A \triangleq \{ A_1, \dots, A_m \}$ is the set of \emph{acceptors}, and
    \item $\mathbf{U} \triangleq \{ U_{P_1}, \dots, U_{P_n}, U_{A_1}, \dots, U_{A_m} \}$ is the set of \emph{utility functions} (also referred to as \emph{preferences}).
\end{itemize}
The proposers and acceptors are disjoint sets of agents representing opposite sides of the market. Each proposer $P_i$ has a utility function $U_{P_i} : \A \cup \{\emptyset\} \to [0 , 1)$ that describes the value they associate with every possible acceptor, and each acceptor $A_j$ has a similar utility function $U_{A_j} : \PP \cup \{\emptyset\} \to [0 , 1)$ over the proposers. We say that proposer $P_i$ \emph{prefers} $A_j$ to $A_k$ if $U_{P_i}(A_j) > U_{P_i}(A_k)$, which we shorthand as $A_j \succ_{P_i} A_k$, and we use the terminology and the notation analogously for the acceptors\footnote{In standard matching market models, agents' preferences are typically represented ordinally. The cardinal utility functions described here are more general, as ordinal rankings can be extracted from them. Thus, we interchangeably use the cardinal and ordinal representations of agents' preferences when convenient.}. We assume that each utility function is injective, meaning no agent is indifferent between any two partners, and we assume that for all $P_i \in \PP$ and $A_j \in \A$, $U_{P_i}(\emptyset) = U_{A_j}(\emptyset) = 0$, indicating that every agent prefers having a partner to not having one, which we write as $\emptyset$.

A \emph{matching} $\mu \subset (\PP \cup \{\emptyset\}) \times (\A \cup \{\emptyset\})$ is a collection of pairs that describes whom each agent is partnered with, where every agent appears in exactly one pair in $\mu$. If $(P_i, A_j) \in \mu$, then $P_i$ and $A_j$ are \emph{matched} with one another. Otherwise, if $(P_i, \emptyset) \in \mu$ or $(A_j, \emptyset) \in \mu$, then $P_i$ or $A_j$ is \emph{unmatched}, respectively. As a slight abuse of notation, we denote the partner of $P_i$ or $A_j$ as $\mu_{P_i}$ or $\mu_{A_j}$, respectively.

The set of all possible matchings $\mathbf{M} \subset 2^{(\PP \cup \{\emptyset\}) \times (\A \cup \{\emptyset\})}$ grows exponentially with the number of agents. In practice, a matchmaker typically seeks to identify from this set a matching that is stable, i.e., one in which no proposer and acceptor would rather be with one another than their assigned partners. We say $P_i$ and $A_j$ form a \emph{blocking pair} for a matching $\mu$ if $A_j \succ_{P_i} \mu_{P_i}$ and $P_i \succ_{A_j} \mu_{A_j}$. Any matching that has a blocking pair is unstable in the sense that agents can deviate from their associated partner and increase their (and their new partner's) utility.
A matching $\mu$ is a \emph{stable matching} if there are no blocking pairs.
It is well known that in every two-sided matching market $(\PP, \A, \mathbf{U})$, the set of stable matchings $\mathbf{M}^{\rm st} \subset \mathbf{M}$ is nonempty \cite{gale1962college}. Furthermore, there exists a unique \emph{acceptor-optimal} stable matching $\mu^* \in \mathbf{M}^{\rm st}$ characterized by the property that for any other stable matching $\mu$, $\mu_{A_j}^* \succeq_{A_j} \mu_{A_j}$ for all $A_j \in \A$, meaning that every acceptor weakly prefers their partner in the acceptor-optimal stable matching to their partner in any other stable matching.

\subsection{Games and Learning}\label{sec:games_learning}

In this section, we describe how the agents interact and form matchings in a dynamic, decentralized market. We model their interactions as a sequence of games where at each timestep $t \in \{1, 2, \dots\}$ a game evolves as follows:

\vspace{.2cm}
\noindent \textbf{Phase 1:} Each proposer $P_i$ selects an action $a_{P_i}^t$ according to a probability distribution $\sigma_{P_i}^t$ over the set of possible actions $\A \cup \{\emptyset\}$. Their action represents the acceptor they choose to propose to, where $a_{P_i}^t = \emptyset$ means that $P_i$ chooses to remain unmatched and does not propose to anyone. We write the collection of proposer actions as $a_\PP^t \triangleq \{a_{P_1}^t, \dots, a_{P_n}^t\}$.

\vspace{.2cm}
\noindent \textbf{Phase 2:} Each acceptor $A_j$ observes the set of proposals they receive, which we denote by $S_{A_j}^t \triangleq \{ P_i \in \PP \ \vert \ a_{P_i}^t = A_j \}$, and selects an action $a_{A_j}^t$ according to a probability distribution $\sigma_{A_j}^t$ over the set of possible actions $S_{A_j}^t \cup \{\emptyset\}$. Their action represents the proposer they choose to accept, where $a_{A_j}^t = \emptyset$ means that $A_j$ chooses to remain unmatched and rejects any and all proposals they receive. We write the collection of acceptor actions as $a_{\A}^t \triangleq \{a_{A_1}^t, \dots, a_{A_m}^t\}$.

\vspace{.2cm}
\noindent \textbf{Phase 3:} The action profile $a^t \triangleq \{a_\PP^t, a_{\A}^t \}$ induces a matching $\mu(a^t)$, where $(P_i, A_j) \in \mu(a^t)$ if and only if $a_{P_i}^t = A_j$ and $a_{A_j}^t = P_i$, and any agent whose action is unreciprocated is unmatched. Each agent then receives the utility that they associate with their partner, i.e., every proposer $P_i$ and acceptor $A_j$ receives utility $u_{P_i}^t(a^t) \triangleq U_{P_i}(\mu_{P_i}(a^t))$ and $u_{A_j}^t(a^t) \triangleq U_{A_j}(\mu_{A_j}(a^t))$, respectively.
\vspace{.2cm}

Here, we assume that when agents select their actions, they know whom they can attempt to match with (i.e., every proposer knows $\A$ and every acceptor knows $\PP$), but they \textbf{do not know their own preferences} over these sets (i.e., they do not know the full form of their utility function). Furthermore, we assume that at no point does any agent observe the action or utility of any other agent (except for each acceptor who observes the actions of those that proposed to them).

\begin{figure}
    \centering
    \includegraphics[width=\linewidth]{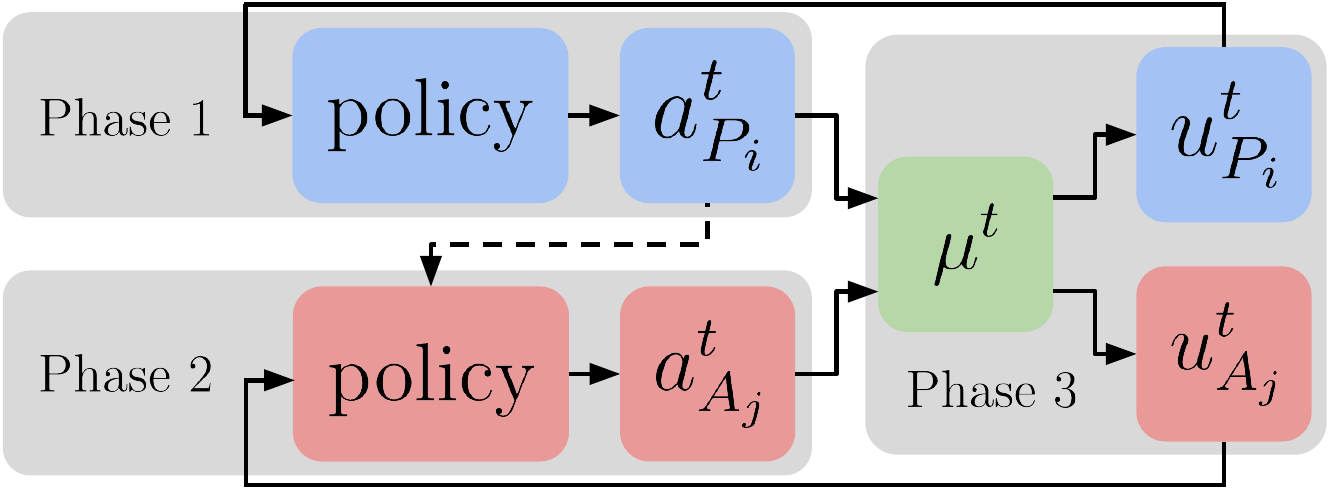}
    \caption{The iterative feedback process through which agents select actions. For simplicity, the diagram explicitly depicts the behavior only of $P_i$ and $A_j$, but the process depends on the behavior of all agents. The dashed line indicates that $A_j$ observes $a_{P_i}^t$ only if $a_{P_i}^t = A_j$.}
    \label{fig:diagram}
\end{figure}


The focus of this work is to develop the current understanding of what can be achieved by learning policies in matching markets. Here, we focus on \emph{completely uncoupled policies} where agents use only their own prior observations of their utilities to select their actions, and ask whether they can
ensure convergence to stable matches.
Formally, we consider policies $\pi_\PP$ and $\pi_\A$ of the form
\begin{align}
    a_{P_i}^t \sim \sigma_{P_i}^t &= \pi_\PP \left(\{ a_{P_i}^\tau, U_{P_i}^\tau \}_{\tau \in \{ 1, 2, \dots \}} \right) \label{eq:proposer_policy} \\
    a_{A_j}^t \sim \sigma_{A_j}^t &= \pi_\A \left(\{ S_{A_j}^\tau, a_{A_j}^\tau, U_{A_j}^\tau \}_{\tau \in \{ 1, 2, \dots \}} \right) \label{eq:acceptor_policy}
\end{align}
for proposer $P_i$ and acceptor $A_j$, respectively. 
This extends the definition of completely uncoupled\footnote{We adopt this definition of completely uncoupled for consistency with the existing literature, but we remark that the level of generality in the information available to the agents and the structure of $\pi$ is unnecessary; the policies presented in this work only require agents to keep track of a limited subset of their prior observations and are stationary. Importantly, there exist impractical trivial policies that either rely on arbitrarily large memories or are non-stationary and can guarantee convergence to any matching. Although these kinds of policies may yield rapid convergence, our results establish that even with limited memory and stationarity, convergence itself is still attainable.} policies \cite{marden2009payoff, pradelski2012learning, marden2014achieving} to the setting where agents make decisions sequentially.



\section{Two-Sided Trial-and-Error Learning}

The existence of completely uncoupled policies that guarantee convergence to equilibrium configurations is well-established in a variety of multi-agent settings, such as potential games, weakly acyclic games, and interdependent games \cite{marden2009payoff, pradelski2012learning, marden2014achieving}). In this section, we seek to understand whether there exist completely uncoupled policies that guarantee convergence to stable matchings in two-sided markets even when there is no central matchmaker and the agents do not know their preferences. Theorem \ref{thm:sm} establishes that these constraints do not preclude stability.

\begin{theorem}\label{thm:sm}
    Consider a matching market $(\PP, \A, \mathbf{U})$. For any $\delta \in (0, 1]$, there exist completely uncoupled policies $\pi_\PP$ and $\pi_\A$ such that $\mathbb{P}[\mu^t \in \mathbf{M}^{\rm st}] > 1 - \delta$ for all sufficiently large $t$.
\end{theorem}

\subsection*{Simple Trial-and-Error Learning Policies}
    The proof of Theorem \ref{thm:sm} is by construction of a pair of policies that are completely uncoupled and provide the desired asymptotic guarantees. In this section we specify these policies, but we defer their formal analysis to the Appendix. We remark that although Theorem \ref{thm:sm} is proven constructively, we have chosen to state it as an existence result as we view the main contribution not as the policies themselves, but rather the guarantees they provide: that is, convergence to stable matches is achievable even in decentralized markets with uncertainty with very simple policies.
    We do not claim that these policies are unique or that they achieve the best convergence rate. Importantly, however, they are indeed completely uncoupled \eqref{eq:proposer_policy}, \eqref{eq:acceptor_policy}, a property that follows from other so-called \emph{trial-and-error learning} policies \cite{pradelski2012learning, shah2024learning} that inspire their construction.
    
    
    First, we introduce proposer trial-and-error learning (PTL), the policy referenced as $\pi_\PP$ in the Theorem statement. PTL relies on a \emph{state} variable $x_{P_i} = \{m, \underline{a}_{P_i}, \underline{u}_{P_i}, \overline{a}_{P_i}, \overline{u}_{P_i}\}$, where
    \begin{itemize}
        \item $m \in \{ C, H, W, D\}$ is their \emph{mood}, which is either content ($C$), hopeful ($H$), watchful ($W$), or discontent ($D$);
        \item $\underline{a}_{P_i} \in \A$ is their \emph{baseline action};
        \item $\underline{u}_{P_i} \in [0, 1)$ is their \emph{baseline utility};
        \item $\overline{a}_{P_i} \in \A$ is their \emph{trial action};
        \item $\overline{u}_{P_i} \in [0, 1)$ is their \emph{trial utility}.
    \end{itemize}
    An action is selected by a proposer $P_i$ using the PTL policy by sampling an acceptor to propose to (parameterized by an \emph{experimentation rate} $\varepsilon$) based on their individual state $x_{P_i}$. Specifically, they select their action according to Algorithm \ref{alg:p_act} (which requires their state) in Phase 1, and their state evolves according to Algorithm \ref{alg:p_update} in Phase 3. Together, these action selection and state update rules form a completely uncoupled policy.
    \renewcommand{\thealgorithm}{1\alph{algorithm}}
    \setlength{\textfloatsep}{1pt}

\begin{figure}
    \centering
    \begin{algorithm}[H]
    \caption{Proposer action selection rule}\label{alg:p_act}
    \begin{algorithmic}
    \Require $x_{P_i}$, $\varepsilon$
    \If{$m_{P_i} \in \{ C, D \}$}
        \State select $a_{P_i} = \emptyset$ with probability (w.p.) $\varepsilon^2$; or
        \State select $a_{P_i} \in \A$ uniformly at random w.p. $\varepsilon$; or
        \State select $a_{P_i} = \underline{a}_{P_i}$ w.p. $1 - \varepsilon - \varepsilon^2$
    \ElsIf{$m_{P_i} = H$}
        \State select $a_{P_i} = \overline{a}_{P_i}$
    \ElsIf{$m_{P_i} = W$}
        \State select $a_{P_i} = \underline{a}_{P_i}$
    \EndIf
    \end{algorithmic}
    \end{algorithm}
    \begin{algorithm}[H]
    \caption{Proposer state update rule}\label{alg:p_update}
    \begin{algorithmic}
    \Require $x_{P_i}$, $a_{P_i}^t$, $u_{P_i}^t$
    \If{$m_{P_i} = C$}
        \If{$a_{P_i}^t = \underline{a}_{P_i}$}
        \State $x_{P_i} = \begin{cases}
                (W, \underline{a}_{P_i}, \underline{u}_{P_i}, \emptyset, 0) & u_{P_i}^t < \underline{u}_{P_i} \\
                (C, \underline{a}_{P_i}, \underline{u}_{P_i}, \emptyset, 0) & u_{P_i}^t = \underline{u}_{P_i} \\
                (H, \underline{a}_{P_i}, \underline{u}_{P_i}, a_{P_i}^t, u_{P_i}^t) & u_{P_i}^t > \underline{u}_{P_i}
            \end{cases}$
        \ElsIf{$a_{P_i}^t \neq \underline{a}_{P_i}$}
        \State $x_{P_i} = \begin{cases}
                (C, \underline{a}_{P_i}, \underline{u}_{P_i}, \emptyset, 0) & u_{P_i}^t \leq \underline{u}_{P_i} \\
                (H, \underline{a}_{P_i}, \underline{u}_{P_i}, a_{P_i}^, u_{P_i}^t) & u_{P_i}^t > \underline{u}_{P_i}
            \end{cases}$
        \EndIf
    \ElsIf{$m_{P_i} = H$}
        \State $x_{P_i} = \begin{cases}
            (D, \emptyset, 0, \emptyset, 0) & u_{P_i}^t \neq \overline{u}_{P_i} \text{ and } \underline{u}_{P_i} = 0 \\
            (C, \underline{a}_{P_i}, \underline{u}_{P_i}, \emptyset, 0) & u_{P_i}^t \neq \overline{u}_{P_i} \text{ and } \underline{u}_{P_i} > 0 \\
            (C, \overline{a}_i, \overline{u}_{P_i}, \emptyset, 0) & u_{P_i}^t = \overline{u}_{P_i}
        \end{cases}$
    \ElsIf{$m_{P_i} = W$}
        \State $x_{P_i} = \begin{cases}
            (D, \emptyset, 0, \emptyset, 0) & u_{P_i}^t < \underline{u}_{P_i} \\
            (C, \underline{a}_{P_i}, \underline{u}_{P_i}, \emptyset, 0) & u_{P_i}^t \geq \underline{u}_{P_i} \\
        \end{cases}$
    \ElsIf{$m_{P_i} = D$}
        \State $x_{P_i} = \begin{cases}
                (D, \emptyset, 0, \emptyset, 0) & u_{P_i}^t = \underline{u}_{P_i} \\
                (H, \underline{a}_{P_i}, \underline{u}_{P_i}, a_{P_i}^t, u_{P_i}^t) & u_{P_i}^t > \underline{u}_{P_i}
            \end{cases}$
    \EndIf
    \end{algorithmic}
    \end{algorithm}
\end{figure}

    Next, we introduce acceptor trial-and-error learning (ATL), the policy referenced as $\pi_\A$ in the Theorem statement. ATL also relies on a \emph{state} variable $x_{A_j} \triangleq (m_{A_j}, \underline{a}_{A_j}, \underline{u}_{A_j}, \underline{S}_{A_j}, \overline{a}_{A_j}, \overline{u}_{A_j})$ maintained by each acceptor $A_j \in \A$, where the added variable $\underline{S}_{A_j}$ is their \emph{baseline proposal set}. The ATL policy provides acceptor $A_j$ their action based on their respective state $x_{A_j}$. Specifically, they select their action according to Algorithm \ref{alg:a_act} in Phase 2, and their state evolves according to Algorithm \ref{alg:a_update} in Phase 3. Together, these rules form a completely uncoupled policy. Notice that unlike PTL, ATL is a purely deterministic policy.

    \setcounter{algorithm}{0}
    \renewcommand{\thealgorithm}{2\alph{algorithm}}

\begin{figure}
    \centering
    \begin{algorithm}[H]
    \caption{Acceptor action selection rule}\label{alg:a_act}
    \begin{algorithmic}
    \Require $x_{A_j}$, $S_{A_j}^t$
    \If{$m_{A_j} = C$}
        \State select $a_{A_j} = \underline{a}_{A_j}$ if $\underline{a}_{A_j} \in S_{A_j}^t \subseteq \underline{S}_{A_j}$; or
        \State select $a_{A_j} = \emptyset$ if $\underline{a}_{A_j} \notin S_{A_j}^t \subset \underline{S}_{A_j}$; or
        \State select $a_{A_j} \in S_{A_j}^t \setminus \underline{S}_{A_j}$ uniformly at random.
    \ElsIf{$m_{A_j} = H$}
        \State select $a_{A_j} = \overline{a}_{A_j}$ if $\overline{a}_{A_j} \in S_{A_j}^t$; or
        \State select $a_{A_j} = \underline{a}_{A_j}$ if $\underline{a}_{A_j} \in S_{A_j}^t$ and $\overline{a}_{A_j} \notin S_{A_j}^t$; or
        \State select $a_{A_j} = \emptyset$ if $\underline{a}_{A_j} \notin S_{A_j}^t$ and $\overline{a}_{A_j} \notin S_{A_j}^t$.
    \ElsIf{$m_{A_j} = W$}
        \State select $a_{A_j} = \underline{a}_{A_j}$ if $\underline{a}_{A_j} \in S_{A_j}^t$; or
        \State select $a_{A_j} = \emptyset$ if $\underline{a}_{A_j} \notin S_{A_j}^t$.
    \ElsIf{$m_{A_j} = D$}
        \State select $a_{A_j} = \emptyset$ if $S_{A_j}^t = \underline{S}_{A_j}$; or
        \State select $a_{A_j} \in S_{A_j}^t \setminus \underline{S}_{A_j}$ uniformly at random.
    \EndIf
    \end{algorithmic}
    \end{algorithm}

    \setcounter{algorithm}{0}
    \renewcommand{\thealgorithm}{2b.\roman{algorithm}}
    \begin{algorithm}[H]
    \caption{Acceptor state update rule}\label{alg:a_update}
    \begin{algorithmic}
    \Require $x_{A_j}$, $a_{A_j}^t$, $u_{A_j}^t$, $S_{A_j}^t$
    \If{$m_{A_j} = C$}
        \State $x_{A_j} = \begin{cases}
            (W, \underline{a}_{A_j}, \underline{u}_{A_j}, S_{A_j}^t, \emptyset, 0) & u_{A_j}^t < \underline{u}_{A_j} \\
            (C, \underline{a}_{A_j}, \underline{u}_{A_j}, S_{A_j}^t, \emptyset, 0) & u_{A_j}^t = \underline{u}_{A_j} \\
            (H, \underline{a}_{P_i}, \underline{u}_{A_j}, S_{A_j}^t, a_{A_j}^t, u_{A_j}^t) & u_{A_j}^t > \underline{u}_{A_j}
        \end{cases}$
    \ElsIf{$m_{A_j} = H$}
        \State $x_{A_j} = \begin{cases}
            (D, \underline{a}_{A_j}, \underline{u}_{A_j}, S_{A_j}^t, \emptyset, 0) & u_{A_j}^t \neq \overline{u}_{A_j} \text{ and } \underline{u}_{A_j} = 0 \\
            (C, \underline{a}_{A_j}, \underline{u}_{A_j}, S_{A_j}^t, \emptyset, 0) & u_{A_j}^t \neq \overline{u}_{A_j} \text{ and } \underline{u}_{A_j} > 0 \\
            (C, \overline{a}_{A_j}, \overline{u}_{A_j}, S_{A_j}^t, \emptyset, 0) & u_{A_j}^t = \overline{u}_{A_j}
        \end{cases}$
    \ElsIf{$m_{A_j} = W$}
        \State $x_{A_j} = \begin{cases}
            (D, \emptyset, 0, \emptyset, \emptyset, 0) & u_{A_j}^t < \underline{u}_{A_j} \\
            (C, \underline{a}_{A_j}, \underline{u}_{A_j}, S_{A_j}^t, \emptyset, 0) & u_{A_j}^t \geq \underline{u}_{A_j} \\
        \end{cases}$
    \ElsIf{$m_{A_j} = D$}
        \State $x_{A_j} = \begin{cases}
            (D, \emptyset, 0, S_{A_j}^t, \emptyset, 0) & u_{A_j}^t = \underline{u}_{A_j} \\
            (H, \emptyset, 0, S_{A_j}^t, a_{A_j}^t, u_{A_j}^t) & u_{A_j}^t > \underline{u}_{A_j}
        \end{cases}$
    \EndIf
    \end{algorithmic}
    \end{algorithm}
\end{figure}

    In the Appendix, we demonstrate that when every proposer follows PTL and every acceptor follows ATL, the stationary distribution of the resulting Markov process converges almost surely to a distribution that supports only stable matchings, which establishes the result.

\section{Exploiting Trial-and-Error Learning}

In this section, we go one step further and evaluate whether the simple learning rules described in the previous section can be exploited. In particular, we consider the setting where the acceptors know that the proposers are using proposer trial-and-error learning, and ask if they can employ an alternate policy that ensures convergence to their most preferred stable matching. In our second contribution, Theorem \ref{thm:aosm}, we show that this kind of exploitation is indeed possible.

\begin{theorem}\label{thm:aosm}
    Consider a matching market $(\PP, \A, \mathbf{U})$, and let $\mu^*$ denote its unique acceptor-optimal stable matching. For every $\delta \in (0, 1]$, there exists $\varepsilon > 0$ and a completely uncoupled policy $\pi_\A^*$ such that if every acceptor follows $\pi_\A^*$ and every proposer follows PTL with experimentation rate $\varepsilon$, then $\mathbb{P}[\mu^t = \mu^*] > 1 - \delta$ for all sufficiently large $t$.
\end{theorem}

\subsection*{Acceptor-Optimal Trial-and-Error Learning Policy}
Like Theorem \ref{thm:sm}, the proof of Theorem \ref{thm:aosm} is constructive. Once again, we state the Theorem as an existence result to emphasize not the policies themselves but rather that equilibrium selection is possible; here, this results from one group adjusting their behavior (in response to the others') to guarantee more favorable asymptotic outcomes. 

Acceptor-optimal trial-and-error learning, or ATL\textsuperscript{*}, is the policy referenced as $\pi_\A^*$ in the Theorem statement. ATL\textsuperscript{*} is formally described by Algorithms \ref{alg:a_act} and \ref{alg:a_update_2}, the latter of which is a mild variant of Algorithm \ref{alg:a_update}. The changes to the algorithm are shown in black; for brevity, we truncate the steps that are the same as in Algorithm \ref{alg:a_update}. Here, the experimentation rate $\varepsilon$ is the same as the one referenced in Algorithm \ref{alg:p_act}, and the functions $F: [0, 1) \to (0.5, 1)$ and $G: [0, 1) \to (0, 0.5)$ are strictly monotonically decreasing.

\renewcommand{\thealgorithm}{2b.\roman{algorithm}}
\begin{figure}[t]
\centering
\begin{algorithm}[H]
\caption{Acceptor-optimal state update rule}\label{alg:a_update_2}
\begin{algorithmic}
    \Require \color{lightgray} $x_{A_j}$, $a_{A_j}^t$, $u_{A_j}^t$, $S_{A_j}^t$, \color{black} $\varepsilon, F, G$
    \If{$m_{A_j} = C$}
        \If{$u_{A_j}^t > \underline{u}_{A_j}$}
            \State $x_{A_j} = \begin{cases}
                (C, \underline{a}_{P_i}, \underline{u}_{A_j}, S_{A_j}^t, \emptyset, 0) & \text{w.p. $1 - \varepsilon^{G(u_{A_j}^t)}$} \\
                (H, \underline{a}_{P_i}, \underline{u}_{A_j}, S_{A_j}^t, a_{A_j}^t, u_{A_j}^t) & \text{w.p. $\varepsilon^{G(u_{A_j}^t)}$}
            \end{cases}$
        \color{lightgray}
        \ElsIf{$u_{A_j}^t \leq \underline{u}_{A_j}$} ...
        \EndIf
    \color{black}
    \ElsIf{$m_{A_j} = D$}
        \If{$u_{A_j}^t > \underline{u}_{A_j}$}
        \State $x_{A_j} = \begin{cases}
            (D, \emptyset, 0, S_{A_j}^t, \emptyset, 0) & \text{w.p. $1 - \varepsilon^{F(u_{A_j}^t)}$} \\
            (H, \emptyset, 0, S_{A_j}^t, a_{A_j}^t, u_{A_j}^t) & \text{w.p. $\varepsilon^{F(u_{A_j}^t)}$}
        \end{cases}$
        \color{lightgray}
        \ElsIf{$u_{A_j}^t = \underline{u}_{A_j}$} ...
        \EndIf
    \ElsIf{$m_{A_j} = H$} ...
    \ElsIf{$m_{A_j} = W$} ...
    \EndIf
\end{algorithmic}
\end{algorithm}
\end{figure}

\color{black}
\vspace{0.2cm}

In the Appendix, we demonstrate that when every proposer follows PTL and every acceptor follows ATL\textsuperscript{*}, the stationary distribution of the resulting Markov process converges almost surely to a distribution that supports only $\mu^*$.


\section{Conclusion}

In this paper, we studied two-sided decentralized matching markets in which neither the proposers nor acceptors knew their own preferences over the opposite side of the market. We presented a model of matching interactions where proposers and acceptors learned about one another over time through limited observations. In our first contribution, we demonstrated that completely uncoupled policies could ensure probabilistic convergence to stable matchings. Then, we demonstrated that the proposer policy was susceptible to exploitation by an alternate acceptor policy that guaranteed convergence to the acceptor-optimal stable matching.

This work introduced novel guarantees for convergence to stability in two-sided markets with unknown preferences. We emphasize that the primary value of our results lies in the guarantees provided by the policies, not the policies themselves. We make no claims about their practical use and instead highlight that they are simple and ensure convergence to stability. We believe such fundamental results can motivate the development of more advanced techniques for achieving stability in real-world markets. Finally, we remark that our second result raises an interesting question about learning in games. Specifically, we showed that in two-sided matching markets, one group of learning agents can improve their outcome by exploiting the policy of the other group. Understanding when exploitation is possible in other non-cooperative games is an exciting avenue for future research.


\bibliographystyle{ieeetr}
\bibliography{references}


\appendix

In this section, we provide proofs of Theorems \ref{thm:sm} and \ref{thm:aosm}. These proofs rely on the theory of \emph{regular perturbed Markov processes}, which we describe briefly below.

\subsection{Regular perturbed Markov processes}

Let $\T^0$ denote the probability transition matrix over a finite state space $\X$. Define $\T^\varepsilon$ as a \lq perturbed' version of $T^0$, where the size of the perturbations are indexed by a scalar $\epsilon > 0$, and let $\T_{x, x'}^\varepsilon$ denote the probability of the transition from $x$ to $x'$. We say that $\T^\epsilon$ is a \emph{perturbed Markov process} if
\begin{enumerate}
    \item $\T^\varepsilon$ is ergodic and aperiodic
    \item $\T_{x, x'}^\varepsilon \to \T_{x, x'}^0$ as $\varepsilon \to 0$ for all $x, x' \in \X$
    \item $\T_{x, x'}^\varepsilon > 0$ for some $\epsilon > 0$ implies that there exists a unique constaint $r(x, x')$ such that $0 < \frac{\T_{x, x'}^\varepsilon}{r(x, x')} < \infty$
\end{enumerate}
We refer to the unique constant $r(x, x')$ that satisfies this third condition as the \emph{resistance} of the transition from $x$ to $x'$. If $\T_{x,x'}^\varepsilon = 0$ for any $\varepsilon > 0$, then $r(x, x') = \infty$, and if $\T_{x, x'}^0 > 0$, then $r(x, x') = 0$.

Any state that is contained in the support of the stationary distribution $\pi^\varepsilon$ of $\T^\varepsilon$ as $\varepsilon \to 0$ is referred to as \emph{stochastically stable}. Equivalently, if we describe the state of a random walk defined by this Markov process at timestep $t$ as $x^t$, then we can say that for every $\delta > 0$, there exist $\epsilon > 0$ and $T \gg 0$ such that $x^t \in \X'$ for all $t > T$ with probability at least $1 - \delta$, where $\X' \subset \X$ denotes the set of all stochastically stable states.

Construct a weighted directed graph $\G$ whose nodes correspond to the recurrence classes of the unperturbed process $\T^0$, and whose edge weights are given by the resistance of the transition between the two nodes they span. A path $x^1, \dots, x^k$ is a sequence of nodes successively connected to one another by directed edges. An \emph{in-tree} $T(x)$ rooted at a node $x$ is a subgraph of $\G$ such that for every node $x'$, there is exactly one path from $x'$ to $x$ in $T(x)$. The \emph{weight} of an in-tree $T(x)$ is the sum of the weights of all of the edges in $T(x)$, and $T^*(x)$ is the weight of the minimum-weight in-tree over all in-trees rooted at $x$.

Equipped with these definitions, we can now introduce a simple criterion for identifying the stochastically stable states of a process \cite{young1993evolution}:

\begin{theorem*}[Young, 1993 \cite{young1993evolution}]
    Let $\T^0$ be a Markov process over a state space $\X$, let $\T^\varepsilon$ be a regular perturbation of $\T^0$, and let $\pi^\varepsilon$ be the stationary distribution over $\X$ for a fixed $\varepsilon$. Then, as $\epsilon \to 0$, $\pi^\varepsilon$ converges to a stationary distribution $\pi^0$ of $\T^0$, and $x \in \X$ is stochastically stable if and only if $x \in \argmin_{x' \in \X} T^*(x')$.
\end{theorem*}

Note that in the proofs below, we construct the graph $\G$ over all of the states $x \in \X$, not just the recurrence classes of $\T^0$, but this does not affect the conclusions (see Lemma 1 in \cite{young1993evolution}).

\newcounter{mycounter}
\setcounter{mycounter}{1}
\newtheorem{claim}{Claim}[mycounter]
\subsection{Proof of Theorem \ref{thm:sm}}

For a given $\varepsilon$, proposer trial-and-error learning defines a regular perturbed Markov process over the state space $\X$, the set of all possible states $x \triangleq \{ x_{P_1}, \dots, x_{P_n}, x_{A_1}, \dots, x_{A_m} \}$. Our goal is to identify the \emph{stochastically stable} states of this process, i.e., the recurrence classes that are visited infinitely often in the limiting stationary distribution induced by the policy as $\varepsilon \to 0$.

We begin by identifying the recurrence classes of the unperturbed Markov process, i.e., the recurrence classes of the process over $\X$ when $\varepsilon = 0$. Let $\Z$ denote the subset of states in which every proposer is either content or discontent. 

\begin{claim}\label{claim:thm1_rec}
    Every state $z \in \Z$ is a singleton recurrence class of the unperturbed process.
\end{claim}

\begin{proof}
    To see this, first observe that any transition out of the state $z \in \Z$ is not possible when $\varepsilon = 0$, since every agent simply repeats their baseline action and receive their baseline utility, causing their state to remain exactly the same. Next, we establish that from every state $x \notin \Z$, there is a transition to a state $z \in \Z$ that occurs with probability 1. Consider a state $x$ in which at least one agent is hopeful or watchful. When in $x$, every discontent, watchful, and content agent plays their baseline action, every hopeful proposer either plays their baseline or trial action, and every hopeful acceptor either plays their baseline, trial, or rejection action. Regardless of the outcome, every hopeful or watchful agent becomes either content or discontent, and their baseline utilities become aligned with the baseline action profiles, resulting in a one-step transition to a state in $\Z$ with probability 1.
\end{proof}

Now, construct a graph $\G$ whose nodes represent each of the states in $\X$. The directed edge from $x$ to $x'$ has weight given by the resistance $r(x, x')$ as described in the previous section. Our goal now is to identify the root nodes of the minimum weight in-trees on the graph $\G$. Let $\E \subset \Z$ denote the subset of states where the baseline action profile $\underline{a}$ corresponds to a stable matching. 

\begin{claim}\label{claim:thm1_w1}
    The minimum of the weights of the out-edges from a node $z \in \Z \setminus \E$ is exactly 1.
\end{claim}

\begin{proof}
    Since the baseline action profile $\underline{a} \in z$ corresponds to an unstable matching, there must be a blocking pair $(P_i, A_j)$. If $P_i$ experimented and proposed to $A_j$ (which occurs with probability $\varepsilon$), then $A_j$ would accept $P_i$, making $P_i$ hopeful and leading to a state in $\X$; in the following timestep, $P_i$ would propose once more to $A_j$, who would also accept them once more, leading to another state in $\Z$ with probability 1. Hence, the total weight along this path is simply 1 + 0 = 1. Furthermore, any transition out of a state $z$ must involve at least one experimentation, so the minimum weight of an outgoing edge must be at least 1, thus completing the proof.
\end{proof}

Next, we recall an important result from \cite{roth1990random}, which states the following: From any unstable matching $\mu^t$, there exists a finite sequence of matchings $\mu^{t+1}, \dots, \mu^{t + k}$ such that for all $l \in \{0, \dots, k - 1\}$, $\mu^{t + l + 1}$ is obtained from $\mu^{t + l}$ by resolving a single blocking pair $(P_{i_l}, A_{j_l})$. Combined with Claim \ref{claim:thm1_w1}, these results imply that for every node $z \in \Z \setminus \E$ in the graph $G$ that represents an unstable matching, there exists a path to a node $e \in \E$ where each edge along the path alternates between having weight 1 (when a proposer first experiments) or weight 0 (when a proposer repeats their hopeful action after experimenting).

Finally, we establish the following:
\begin{claim}\label{claim:thm1_w2}
    The total weight of a path from a node $e \in \E$ to another node $z \notin \E$ must be at least 2.
\end{claim}

\begin{proof}
    Since the acceptors always accept any new proposal they receive, a unilateral proposer experimentation in $e$ leads to a state $x \in \X \setminus \Z$ in which at least one agent is watchful or hopeful. However, since the original baseline action $\underline{a}$ is stable, at least one of the two agents must prefer their previous partner to one another. In either case, there is a zero-resistance path that leads from $x$ back to $e$ in at most two transitions. Hence, any such transition out of $e$ must have resistance strictly greater than 1. There are three possible transitions out of a stable matching that occur with resistance 2. First, two proposers can experiment simultaneously and cause a transition to another matching (e.g., by swapping partners); from Algorithm \ref{alg:a_act}, the acceptors tentatively accept the new proposers, leading to a state $x \in \X$; then, depending on their preferences, this leads to a new state $z \in \Z \setminus \E$ (where $z \neq e)$ in which they are either content or discontent. Second, two proposers can experiment and propose to the same acceptor $A_j$ in succession, causing the content proposer $\mu_{A_j}$ to become discontent; the procedure is similar as the one described above. Third, any content proposer $P_i$ can play $a_{P_i} = \emptyset$; this will cause their partner $\mu_{P_i}$ to become watchful then discontent, causing a two-timestep transition to a state $z \in \Z \setminus \E$. While the first two transitions may not always possible (for example, in a stable matching where every proposer is partnered with their unique top choice acceptor), the third transition is always possible and occurs with resistance exactly equal to 2.
\end{proof}

The previous claims establish that the minimum weight of an out-edge from a state $x \in \X \setminus \Z$ is 0, the minimum weight of an out-edge from a state $z \in \Z \setminus \E$ is 1, and the minimum weight of a path from a state $e \in \E$ to a state $z \in \Z \setminus \E$ is 2. We complete the proof by establishing the following:

\begin{claim}\label{claim:thm1_tree}
    A minimum weight in-tree on the graph $\G$ must be rooted at a state $e \in \E$.
\end{claim}

\begin{proof}
    We proceed by contradiction. First, consider a state $x \in \X \setminus \Z$, and suppose that $T(x)$ is a minimum weight in-tree rooted at $x$ with total weight $\rho(x)$. Let $z$ be the node that is reached in one timestep when all agents repeat their baseline actions in $x$. Since $T(x)$ is a tree, there must be exactly one outgoing edge from $z$, and by Claim \ref{claim:thm1_w1}, this edge must have weight at least 1. Construct a new in-tree rooted at $z$ by removing this outgoing edge from $z$, and adding the edge from $x$ to $z$ which has weight 0. This in-tree has total weight $\rho(x) - 1$, implying a contradiction.
    
    Next, consider a state $z \in \Z \setminus \E$, and suppose that $T(z)$ is a minimum weight tree rooted at $z$ with total weight $\rho(z)$. Let $z, x^1, z^1, \dots, x^k, z^k$ be a sequence of states from $z$ to some state $z^k \in \E$ corresponding to a stable matching, where each pair $(x^l, z^l)$ is obtained from the previous by resolving some blocking pair. We consider two cases:
    
    \noindent \textit{Case 1:} Suppose that the outgoing edge from $e$ has weight greater than or equal to 2. Construct a new in-tree rooted at $z^k$ by removing all of the existing outgoing edges from $x^1, z^1, \dots, x^k, z^k$, adding the new edges that follow this path sequentially as described above, and adding the edge from $z$ to $x^1$. The total weight of the new edges added along the path from $x^1$ to $z^k$ is exactly $k - 1$, and by Claims \ref{claim:thm1_rec} and \ref{claim:thm1_w1}, the total weight along this path before must also have been at least $k - 1$. In this case, since the outgoing edge from $z$ to $x^1$ has weight exactly equal to one, the new in-tree must have weight at most $\rho(z) - 1$, implying a contradiction.
    
    \noindent\textit{Case 2:} Suppose that the outgoing edge from $e$ has weight exactly equal to one. From the discussion in Claim \ref{claim:thm1_w2}, this must correspond to the case where exactly one proposer experiments in $e$, meaning that the outgoing edge from $e$ state $x \in \X \setminus \Z$ in one transition. Construct a new in-tree rooted at $z^k$ by removing all of the existing outgoing edges from $x^1, z^1, \dots, x^k, z^k$, adding the new edges that follow this path sequentially as described above, adding the edge from $z$ to $x^1$, and adding the edge from $x$ to $e$. By Claim \ref{claim:thm1_rec}, the unique minimum weight edge out of $x$ leads back to $e$; however, because $T(z)$ is a tree, this edge could not have been a part of $T(z)$, since the edge from $e$ to $x$ was already part of $T(z)$. Thus, the outgoing edge from $x$ in $T(z)$ must have resistance at least $1$. Furthermore, since the resistance of the outgoing edge from $e$ has weight exactly 1, the total resistance of this new tree rooted at $z^k$ is at most $\rho(z) - 1$, which establishes another contradiction. We conclude that the only stochastically stable states of the process must belong to $\E$.
\end{proof}

\stepcounter{mycounter}
\subsection{Proof of Theorem \ref{thm:aosm}}

Several components of the proof of Theorem \ref{thm:sm} extend to the proof of Theorem \ref{thm:aosm}. The Markov process induced by proposer and acceptor-optimal trial-and-error learning transitions over the same state space $\X$, and the recurrence classes of the unperturbed process are still the singleton states in $\Z$. The minimum resistance of a transition out of a state $e$ remains exactly 2 and occurs as described in Claim \ref{claim:thm1_w2}. However, the minimum resistance transitions from states $z \in \Z \setminus \E$ are now slightly different.

\begin{claim}\label{claim:thm3_w1}
    The minimum of the weights of the out-edges from a node $z \in \Z \setminus \E$ is strictly greater than 1 and less than 2.
\end{claim}

\begin{proof}
    Recall that a transition out of such a state $z$ involves resolving a blocking pair $(P_i, A_j)$. If acceptor $A_j$ is content at timestep $t$ and they receive a new proposal that yields utility $u_{A_j}^t > \underline{u}_{A_j}$, then they become hopeful with resistance $G(u_{A_j}^t)$. If acceptor $A_j$ is discontent, indicating that they are unmatched, then they do the same with resistance $F(u_{A_j}^t)$. Thus, for a given state $z$, the minimum weight out-edge has weight $1 + G(u_{A_j})$ if there exists a blocking pair $(P_i, A_j)$ where $A_j$ is matched, or $1 + F(u_{A_j})$ otherwise; here, the $1$ comes from the fact that a proposer must still experiment to cause this deviation. The maximum of the ranges of the function $F$ and $G$ is strictly less than $1$, so the total minimum edge weight is strictly less than $2$.
\end{proof}

These transitions replicate so-called \emph{best response dynamics} in matching markets with full information \cite{ackermann2008uncoordinated}, which proceed in two phases: In each round of the first phase, one matched proposer plays a best response; if the matching is unstable, this involves resolving their most preferable blocking pair. In each round of the second phase (which begins after all of the blocking pairs involving matched proposers have been resolved), one unmatched proposers plays a best response. It is known that from any initial unstable matching $\mu$, every possible sequence of matchings obtained from the best response dynamics terminates at a stable matching \cite{ackermann2008uncoordinated}. Conversely, one could define a similar process whereby each matched acceptor is matched with their most preferred proposer in the first phase, followed by a second phase in which each unmatched acceptor does the same. It follows immediately from the proof in \cite{ackermann2008uncoordinated} that such a process would also terminate at a stable matching. We refer to this process as \emph{acceptor best-response dynamics}.

\begin{claim}\label{claim:thm3_brd}
    Consider any state $z \in \Z \setminus \E$, and let $z, x^1, z^1, \dots, x^k, z^k$ be any sequence of states such that the resistance of each transition in the sequence is minimum over all possible transitions. Then, the baseline action profiles of the subsequence $z, z^1, \dots, z^k$ correspond to the acceptor best-response dynamics, and $z^k$ corresponds to a stable matching. 
\end{claim}

\begin{proof}
    From Claim \ref{claim:thm3_w1}, the minimum resistance transition out of any state $z$ corresponds to the situation where a proposer experiments and is matched with a new acceptor, who accepts them probabilistically according to a function that depends on whether they are matched ($G$) or unmatched ($F$).
    The range of the function $F$ is strictly greater than the range of $G$, which implies that in a given state $z$ whose corresponding matching has multiple blocking pairs, a blocking pair consisting of a matched acceptor is more likely to be resolved than a blocking pair consisting of an unmatched acceptor. Furthermore, if an acceptor participates in multiple blocking pairs, they are most likely to resolve the blocking pair that yields the greatest utility, since $F$ and $G$ are both strictly monotonically decreasing. From these statements, it is clear that the sequence $z, x^1, z^1, \dots, x^k, z^k$ can be divided into two subsequences $z, x^1, z^1, \dots, x^i, z^i$ and $x^{i+1}, z^{i + 1}, \dots, x^k, z^k$, where every transition in the first (second) sequence corresponds to the resolution of a blocking pair involving a matched (unmatched) acceptor. This aligns precisely with the acceptor best-response dynamics, and thus must also terminate at some state $z^k$ corresponding to a stable matching.
\end{proof}

Now, we introduce another important definition. If $\mu$ is a stable matching, we say that $\mu_{-P_i}$ is the \emph{near-stable} matching with respect to $\mu$ resulting from making proposer $P_i$ and their partner $\mu_{P_i}$ single. It is known that if the proposers follow the best response dynamics starting from a near-stable matching $\mu_{-P_i}$, then the stable matching $\mu'$ reached at the end of the process satisfies $\mu_{P_i}' \succeq_{P_i} \mu_{P_i}$ for all $P_i \in \PP$ \cite{shah2024learning}. Following similar reasoning, it is easily verifiable that every sequence of acceptor-best response dynamics starting from a near-stable matching $\mu_{-A_j}$ terminates a stable matching that satisfies $\mu_{A_j}' \succeq_{A_j} \mu_{A_j}$ for all $A_j \in \A$.

Let $\G'$ denote the graph over $\X$ with edge weights corresponding to the resistance of the transitions induced by these policies. The same reasoning from the proofs of \ref{thm:sm} can be applied to show that a minimum weight tree can only be rooted at a node $e \in \E$, so we do not repeat the arguments here. Let $e^* \in \E$ denote the state corresponding to the acceptor-optimal stable matching. To complete the proof, we must establish the following:

\begin{claim}\label{claim:thm3_tree}
    The total weight of a tree rooted at a node $e \in \E \setminus \{e^*\}$ is strictly greater than the total weight of a tree rooted at $e^*$.
\end{claim}

\begin{proof}
    We begin with an explicit construction of a tree rooted at $e^*$. Initialize $T(e^*)$ as an empty set of edges, and add to it as follows. Consider a node $z \in \Z \setminus \E$. If $z$ already has an outgoing in $T(e^*)$, then do nothing. Otherwise, add a minimum weight edge from $z$ to some other node $z^1$. Continue in this fashion, adding nodes $z^2, z^3, \dots$ until either reaching a state $z^k \in T(e^*)$ or until reaching a state $z^k \in \E$. Once all states $z \in \Z \setminus \E$ have been considered, for each state $e \neq e^*$, add an outgoing edge with weight 2 to a node whose baseline action corresponds to a near-stable matching. By construction, the resulting graph $T(e^*)$ is an in-tree rooted at $e^*$ with the property that the weight of the unique outgoing edge from every node is minimum among the weights of all possible outgoing edges from that node. Let $\rho(e^*)$ denote the total weight of this tree.
    
    Now, suppose that $T(e)$ is a tree rooted at $e \neq e^*$ with total weight $\rho(e)$. Since $T(e)$ is a tree, there exists a unique path from $e^*$ to $e$ in $T(e)$. The outgoing edge from $e^*$ must lead to a state corresponding to a near-stable matching. By Claim \ref{claim:thm3_brd}, any path consisting solely of minimum resistance transitions must correspond to the acceptor best-response dynamics. Additionally, every sequence of acceptor best-response dynamics initiated from a near-stable matching with respect to $\mu^*$ must lead back to $\mu^*$. It necessarily follows that at least one transition along the path from $e^*$ to $e$ has non-minimum resistance. Given that the weight of the outgoing edge from $e$ is at least 2 and the weight of the outgoing edge from $e^*$ is 2, and given that the weight of the outgoing edge from every node in $T(e)$ is at least the weight of the corresponding outgoing edge in $T(e^*)$ (by construction), it follows that $\rho(e) > \rho(e^*)$.
\end{proof}

\end{document}